\documentclass[11pt,letterpaper]{article}
\usepackage[utf8]{inputenc}
\pdfoutput=1 %requested by arxiv
\usepackage[letterpaper,margin=1in]{geometry}

\usepackage{amsmath,amssymb,amsthm,amsfonts,latexsym,bbm,xspace,thm-restate}
\usepackage{graphicx,float,mathtools,braket,slantsc}
\usepackage{enumitem,booktabs,forest,mathdots,soul}
\usepackage[useregional]{datetime2}
\DTMusemodule{english}{en-US}

\usepackage[colorlinks,citecolor=blue,,linkcolor=magenta,bookmarks=true]{hyperref}
\usepackage[capitalise]{cleveref}

\usepackage{multicol,array}

\newtheorem{theorem}{Theorem}
\newtheorem{lemma}[theorem]{Lemma}

\newtheorem{corollary}[theorem]{Corollary}
\newtheorem{conjecture}[theorem]{Conjecture}

\crefname{result}{Result}{Results}

\crefname{fact}{Fact}{Facts}

\theoremstyle{definition}
\newtheorem{definition}[theorem]{Definition}

\DeclareMathOperator{\lambdaop}{\lambda}
\DeclareMathOperator{\Deltaop}{\Delta}

\newcommand{\ketbra}[2]{\ket{#1}\!\bra{#2}}

\newcommand{\YES}{\text{YES}}
\newcommand{\NO}{\text{NO}}

\newcommand{\class}[1]{\ensuremath{\mathsf{#1}}\xspace}
\mathchardef\mhyphen="2D % Define hyphen for complexity classes

\newcommand{\PTIME}{\class{P}}

\newcommand{\coNP}{\class{coNP}}

\newcommand{\QMA}{\class{QMA}}

\newcommand{\UQMA}{\class{UQMA}}

\newcommand{\PQMAlog}{\PTIME^{\QMA[\log]}}
\newcommand{\PUQMAlog}{\PTIME^{\UQMA[\log]}}

\newcommand{\PQMApar}{\PTIME^{||\QMA}}

\newcommand{\prob}[1]{\textup{\textsc{#1}}\xspace}

\newcommand{\kLH}[1][k]{\ensuremath{#1}\prob{-LH}}

\newcommand{\spectralGap}{\prob{Spectral Gap}}

\newcommand{\calX}{\mathcal{X}}
\newcommand{\calY}{\mathcal{Y}}

\DeclarePairedDelimiter\parend{\lparen}{\rparen}

\newcommand{\paren}[1]{\parend*{#1}}

\newcommand{\bO}[1]{\operatorname*{O}\paren{#1}}

\newcommand{\ie}{i.e.\xspace}

\newcommand{\eg}{e.g.\xspace}

\title{
    A Note on the Complexity of the Spectral Gap Problem
}

\author{
   Justin Yirka\\\small{\textsl{The University of Texas at Austin}}\\\small{\texttt{\href{mailto:yirka@utexas.edu}{yirka@utexas.edu}}}
}

\hypersetup{
    pdftitle={A Note on the Complexity of the Spectral Gap Problem},
    pdfauthor={Justin Yirka},
}

\date{March 2025}

\begin{document}

\maketitle

\begin{abstract}
    The problem of estimating the spectral gap of a local Hamiltonian is known to be
    contained in the class $\PQMAlog$: polynomial time with access to a logarithmic number
    of $\QMA$ queries.
    The problem was shown to be hard for $\PUQMAlog$, a weaker class, under Turing reductions
    by Gharibian and Yirka (Quantum 3, 189 (2019)).

    I give a brief proof that the Spectral Gap problem is $\QMA$-hard under a many-one (Karp) reduction.
    Consequently,
    the problem is
    $\PQMAlog$-complete under truth-table reductions.

    It remains open to characterize the complexity of the Spectral Gap problem under
    many-one reductions.
    I conjecture that the problem belongs to a strict subclass of $\PQMAlog$.
\end{abstract}
\smallskip

The spectral gap of a local Hamiltonian is the difference between the first and second eigenvalues.
The spectral gap is an important quantity in quantum computation, most notably because
it implies a lower bound on the runtime of the adiabatic algorithm
\cite{farhi2000quantumcomputationadiabaticevolution}.
It also affects the complexity of estimating the ground state energy
\cite{FeffermanAbhinav_DGF22_spectralGap}.

Let $\lambdaop_1(H) \leq \lambdaop_2(H) \leq \dots \leq \lambdaop_N(H)$
denote the eigenvalues of $H$ in non-decreasing order,
and let $\Deltaop\paren{H} = \lambdaop_2\paren{H} - \lambdaop_1\paren{H}$ denote the spectral gap of $H$.
To study the complexity of estimating this value,
Ambainis \cite{Ambainis14-hardThanQMA} defined the \spectralGap{} problem.\footnote{
    This definition allows for the possibility that $\Deltaop\paren{H} = 0$.
    Some works use ``spectral gap'' to refer to the difference between the first two \emph{distinct} eigenvalues.
    This choice is significant:
    in the adiabatic algorithm, if all states in a degenerate ground space are acceptable solutions,
    then it is actually this ``excited'' spectral gap that lower bounds the algorithm's run time.
    However, estimating the excited spectral gap appears very difficult.
    For example, consider a Hamiltonian with a degenerate ground space
    and suppose a prover offers two states to a verifier.
    For the verifier to
    distinguish between two states with the same energy (which do not define the excited gap, which may be small or large)
    and two states with exponentially close energies (forming a small excited gap)
    would likely be intractable.}\textsuperscript{,}\footnote{
    This problem is distinct from the study of families of transitionally invariant
    Hamiltonians in the thermodynamic limit on increasing lattice sizes.
    Estimating the spectral gap in this context was shown to be undecidable \cite{undecidability-nature}.
    Our problem takes a particular Hamiltonian on a particular system size as input.
}
He showed the problem is contained in the class $\PQMAlog$.

\begin{definition}[\spectralGap]\label{def:spectralGap}
    Given a $k$-local Hamiltonian $H$
    and real numbers $a,b$ such that $b-a\geq n^{-c}$,
    for $n$ the number of qubits $H$ acts on and $c>0$ some constant, decide:
    \begin{itemize}[nosep]
        \item If $\Deltaop\paren{H}\leq a$, output \YES{}.
        \item If $\Deltaop\paren{H}\geq b$, output \NO{}.
    \end{itemize}
\end{definition}

\begin{definition}
    $\PQMAlog$ is the class of problems decidable by a polynomial-time deterministic Turing machine making
    $\bO{\log n}$ queries to an oracle for $\QMA$.\footnote{
        Because $\QMA$ is a class of promise problems,
        the $\PTIME$ machine in $\PQMAlog$ might make invalid queries that violate the promise.
        For example, a query to the $\QMA$-complete problem $\kLH$ asks whether the ground state energy
        of a given Hamiltonian is at most $a$ or at least $b$ for parameters $a,b$,
        \emph{promised} that one is the case, meaning it is not between $a$ and $b$.
        Following the standard convention \cite{Goldreich06-promiseProblems},
        the oracle may answer any invalid queries arbitrarily and the $\PTIME$ machine must output the same final answer
        no matter how any invalid queries are answered.
        See \cite{GharibianYirka-GY19-PQMAlog-APXSIM} for a longer discussion.
    }
\end{definition}

\begin{lemma}[\cite{Ambainis14-hardThanQMA}]\label{sg-contained}
    \spectralGap{} is contained in $\PQMAlog$.
\end{lemma}
\begin{proof}[Proof sketch]
    Given a Hamiltonian $H$,
    query the $\QMA$ oracle with instances of $\kLH$
    to perform binary search and
    estimate the ground state energy to within inverse-polynomial precision;
    use the slightly modified binary search algorithm of \cite{Ambainis14-hardThanQMA} that is robust to oracles for promise problems.
    Similarly, given that estimating $\lambdaop_c\paren{H}$ is $\QMA$-complete for any constant $c$ \cite{JordanGossetLove-JGL10-excitedkLH},
    binary search for $\lambda_2\paren{H}$.
    This gives an estimate of $\Deltaop\paren{H}$ using at most $\bO{\log n}$ queries.
\end{proof}

Ambainis \cite{Ambainis14-hardThanQMA} also claimed to show that \spectralGap{} is $\PUQMAlog$-hard.
$\UQMA$, for ``Unique''-$\QMA$, is the subclass of $\QMA$ in which, in the \YES{} case,
the subspace of accepting witnesses is 1-dimensional.
Gharibian and Yirka \cite{GharibianYirka-GY19-PQMAlog-APXSIM} later identified a bug in the hardness proof
related to handling invalid queries.
They gave an alternative proof that \spectralGap{} is $\PUQMAlog$-hard under a polynomial-time Turing reduction (Cook reduction).
Specifically, their approach decided any $\PUQMAlog$ problem using a polynomial number of queries to \spectralGap{}.

In this note, I improve on the result of \cite{GharibianYirka-GY19-PQMAlog-APXSIM}.
I show that \spectralGap{} is $\QMA$-hard under a polynomial-time many-one reduction (Karp reduction).
This trivially implies $\PQMAlog$-hardness under a polynomial-time Turing reduction making a logarithmic number of calls to \spectralGap{}.
Because this reduction uses a logarithmic number of queries, instead of polynomial,
it implies the problem is $\PQMAlog$-hard under a truth-table reduction.

\section*{Results}
\begin{lemma}\label{sg-QMA-hard}
    \spectralGap{} is $\QMA$-hard under a polynomial-time many-one reduction.
\end{lemma}
\begin{proof}
    Consider an arbitrary instance of the $\QMA$-complete problem $\kLH[2]$
    with Hamiltonian $H$ and parameters $a=1/3,b=2/3$
    with the goal of deciding
    whether $\lambdaop_1\paren{H} \leq a$ or $\lambdaop_1\paren{H} \geq b$.
    We will show that $\kLH[2]\paren{H,a,b}$ is equivalent to
    $\spectralGap\paren{H',a,b}$ for $H'$ constructed below.

    Denote the space on which $H$ acts as $\calY$. Introduce an additional single-qubit space $\calX$,
    and construct the 3-local Hamiltonian
    \[
        H' = \ketbra{0}{0}_{\calX} \otimes (I - \ketbra{0\dots0}{0\dots0})_{\calY} + \ketbra{1}{1}_{\calX} \otimes H_{\calY} .
    \]

    $H'$ is block-diagonal with respect to $\calX$, and we can analyze the spectrum of $H'$ accordingly.
    In the $\ket{0}$-block,
    there is a single nullstate
    $\ket{\phi_0} := \ket{0}_{\calX}\ket{0\dots0}_{\calY}$.
    All other states in the $\ket{0}$-block have have energy 1.
    In the $\ket{1}$-block, the spectrum of $H'$ is the same as $H$.

    Since $H$ is positive semidefinite (by definition of $\kLH$),
    no eigenvalues in the $\ket{1}$-block are below the
    nullstate in the $\ket{0}$-block, so $\lambdaop_1\paren{H'}= 0$.
    To find the second eigenvalue,
    we take the minimum of the eigenvalues from the $\ket{1}$-block and
    the remaining eigenvalues from the $\ket{0}$-block, which are all 1.
    So $\lambdaop_2\paren{H'} = \min\{ \lambdaop_1\paren{H}, 1 \}$.
    Therefore the spectral gap is
    \[
        \Deltaop\paren{H'} = \lambdaop_2\paren{H'}-\lambdaop_1\paren{H'} = \min\{ \lambdaop_1(H), 1 \} .
    \]

    If $\lambdaop_1\paren{H} \leq a$, then $\Deltaop\paren{H'} \leq 1/3 = a$.
    Inversely, if $\lambdaop_1\paren{H} \geq b$, then $\Deltaop\paren{H'} = \geq 2/3 = b$.
    Therefore, the answer to the $\kLH[2]$ instance is the same as the \spectralGap{} instance.
\end{proof}

$\PQMAlog$-hardness under Turing reductions follows trivially.
The $\PQMAlog$ computation can be simulated in polynomial-time,
substituting any query to $\QMA$ with a query to \spectralGap{}.

\begin{corollary}\label{sg-hard-turing}
    \spectralGap{} is $\PQMAlog$-hard under a polynomial-time Turing reduction.
\end{corollary}

Because the reduction for \cref{sg-hard-turing} makes one query to \spectralGap{} for each query to $\QMA$,
it makes at most a logarithmic number of queries.
As I explain below, this implies hardness under truth-table reductions,
a stricter type of reduction than Turing reductions.

A polynomial-time truth-table reduction allows for a polynomial number of queries
but requires all queries to be non-adapative, \ie{} in parallel
(see \eg{} \cite{Beigel91-boundedQueriesBooleanHierarchy,BussHay-BH91-truthTableSAT}).
In contrast, general Turing reductions are allowed to make queries,
process answers,
and base future queries on those answers.

\begin{theorem}\label{sg-hard-truthtable}
    \spectralGap{} is $\PQMAlog$-hard under a polynomial-time truth-table reduction.
\end{theorem}
\begin{proof}
    Beigel \cite{Beigel91-boundedQueriesBooleanHierarchy} proved
    that for all complexity classes $\class{X}$,
    a logarithmic number of adapative queries can be reduced to a polynomial number of non-adaptive queries.
    In notation,
    $\PTIME^{\class{X}\left[\log\right]} \subseteq \PTIME^{\class{||X}}$ for all classes $\class{X}$,
    where the latter class uses non-adaptive queries.
    Indeed, \cite{GPY} showed that $\PQMAlog=\PQMApar$,
    where the forward containment follows from Beigel's result.

    The reduction is as follows.
    With $\bO{\log n}$ queries, there are $2^{\bO{\log n}}$ possible sequences of query answers.
    Therefore, the $\PTIME^{\class{X}\left[\log\right]}$ computation can be simulated over all paths in polynomial time,
    generating at most a polynomial number of different queries it might ask the $\class{X}$ oracle.
    After pre-computing these possible queries,
    a $\PTIME^{\class{||X}}$ machine can ask all of them in parallel.
    The answers reveal the correct computation path and the corresponding output.

    Using the above approach,
    the polynomial-time Turing reduction of $\PQMAlog$ to \spectralGap{}
    making a logarithmic number of queries
    can be converted into a polynomial-time reduction making a polynomial number
    of non-adaptive queries.
    This is a truth-table reduction, as desired.
\end{proof}

Finally, combining \cref{sg-contained,sg-hard-truthtable} yields the following theorem.

\begin{theorem}
    \spectralGap{} is $\PQMAlog$-complete, with hardness under polynomial-time truth-table reductions $(\PQMAlog \leqslant^{\mathsf{p}}_{\mathsf{tt}} \spectralGap{})$.
\end{theorem}

\subsection*{Open Problem}
This note improves the characterization of \spectralGap{}
from $\PUQMAlog$-hard under polynomial-query Turing (Cook) reductions
to $\PQMAlog$-hard under a truth-table reduction.
This matches the containment of \spectralGap{} in $\PQMAlog$.
It remains open to show completeness under a polynomial-time many-one (Karp) reduction.

For bounded-query classes like $\PQMAlog$,
the difference between Turing and many-one reductions seems especially meaningful.
Pointedly,
although it is believed $\QMA$ is a \emph{strict} subclass of $\PQMAlog$ (or else $\QMA$ equals $\class{coQMA}$),
the $\QMA$ problem $\kLH$ is $\PQMAlog$-hard under truth-table reductions!
Normally, showing a problem is complete for a complexity class
can be interpreted as saying the problem characterizes and is characterized by the class.
It seems truth-table reductions cannot be interpreted as strongly in this context.
(In contrast, the problem \prob{Apx-Sim} is known to be $\PQMAlog$-complete
under a many-one reduction \cite{Ambainis14-hardThanQMA,GharibianYirka-GY19-PQMAlog-APXSIM}.)

Additionally, the class $\PQMAlog$, like $\PTIME$,
is closed under complement, whereas
the \YES{} and \NO{} cases in
\spectralGap{} appear asymmetrical
(the $\PQMAlog$-complete problem \prob{Apx-Sim} was eventually shown to be equivalent with a symmetrical definition, \prob{$\forall$-Apx-Sim} \cite{GPY}).
Perhaps instead of attempting to prove \spectralGap{} is $\PQMAlog$-hard,
the correct goal is to prove a stricter upper bound on the problem.

\begin{conjecture}
	\spectralGap{} is \emph{not} $\PQMAlog$-hard under many-one reductions,
	and the problem is contained in a weaker class than $\PQMAlog$.
\end{conjecture}

Otherwise, to prove \spectralGap{} is $\PQMAlog$-hard under many-one reductions,
a good starting point would be to show $\class{coQMA}$- or even $\coNP$-hardness.

\subsection*{Acknowledgments}
This material is based upon work supported by the U.S. Department of Energy, Office of Science, National Quantum Information Science Research Centers, Quantum Systems Accelerator.

\bibliographystyle{alphaurl}
\bibliography{bibliography}

\end{document}